\documentclass[journal]{IEEEtran}

\usepackage[utf8]{inputenc}
\usepackage{color}
\usepackage{xcolor}
\usepackage{array}
\usepackage{verbatim}
\usepackage{float}
\usepackage{amsmath}
\usepackage{amsthm}
\usepackage{amssymb}
\usepackage{graphicx}
\usepackage{longtable}
\usepackage{multirow}
\usepackage{booktabs}
\usepackage[unicode=true,
bookmarks=false,
breaklinks=false,pdfborder={0 0 1},colorlinks=false]
{hyperref}
\hypersetup{
	colorlinks,bookmarksopen,bookmarksnumbered,citecolor=blue,urlcolor=blue}
\usepackage{cite}

\usepackage{lipsum}
\usepackage{mathtools}
\usepackage{cuted}

\usepackage{algorithmic}
\usepackage{longtable}
\usepackage{balance}

\floatstyle{ruled}
\newfloat{algorithm}{tbp}{loa}
\providecommand{\algorithmname}{Algorithm}
\floatname{algorithm}{\protect\algorithmname}

\makeatletter
\let\oldforeign@language\foreign@language
\DeclareRobustCommand{\foreign@language}[1]{%
	\lowercase{\oldforeign@language{#1}}}

\let\oldforeign@language\foreign@language
\DeclareRobustCommand{\foreign@language}[1]{%
	\lowercase{\oldforeign@language{#1}}}

\ifCLASSINFOpdf
\else
\fi

\@ifundefined{showcaptionsetup}{}{%
	\PassOptionsToPackage{caption=false}{subfig}}
\usepackage{subfig}

\usepackage{balance}

\ifCLASSINFOpdf
\else
\fi

\newtheorem{lem}{Lemma}

\newtheorem{thm}{Theorem}
\newtheorem{rem}{Remark}

\newtheorem{assum}{Assumption}
\ifCLASSINFOpdf
\else
\fi

\def\ps@IEEEtitlepagestyle{%
	\def\@oddhead{\parbox[t][\height][t]{\textwidth}{\centering \scriptsize
			Personal use of this material is permitted. Permission from the author(s) and/or copyright holder(s), must be obtained for all other uses. Please contact us and provide details if you believe this document breaches copyrights.\\
			\noindent\makebox[\linewidth]{}
		}\hfil\hbox{}}%
	\def\@evenhead{\scriptsize\thepage \hfil \leftmark\mbox{}}%
	\def\@oddfoot{\parbox[t][\height][l]{\textwidth}{
			\vspace{-20pt}{\rule{\textwidth}{0.4pt}}\\ \footnotesize			{\bf{\footnotesize\textcolor{red}{A. Shevidi and H. A. Hashim, "Adaptive Backstepping and Non-singular Sliding Mode Control for Quadrotor UAVs with Unknown Time-varying Uncertainties," in Proceedings of the 2024 American Control Conference (ACC), Ottawa, Canada, 2024.}}}\\
			\noindent\makebox[\linewidth]
		}\hfil\hbox{}}%
	\def\@evenfoot{\MYfooter}}

\makeatother
\begin{document}
	\bstctlcite{IEEEexample:BSTcontrol}

\title{Adaptive Backstepping and Non-singular Sliding Mode Control for Quadrotor UAVs with Unknown Time-varying Uncertainties}

\author{Arezo Shevidi and Hashim A. Hashim
	\thanks{This work was supported in part by National Sciences and Engineering
		Research Council of Canada (NSERC), under the grants RGPIN-2022-04937
		and DGECR-2022-00103.}
	\thanks{A. Shevidi and H. A. Hashim are with the Department of Mechanical
		and Aerospace Engineering, Carleton University, Ottawa, Ontario, K1S-5B6,
		Canada, (e-mail: hhashim@carleton.ca).}
}



\maketitle
\begin{abstract}
	This paper presents a novel quaternion-based nonsingular control system for underactuated vertical-take-off and landing (VTOL) Unmanned Aerial Vehicles (UAVs). Position and attitude tracking is challenging regarding singularity and accuracy. Quaternion-based Adaptive
Backstepping Control (QABC) is developed to tackle the underactuated issues of UAV control systems in a cascaded way.
Leveraging the virtual control (auxiliary control) developed in the QABC,
desired attitude components and required thrust are produced. Afterwards, we
propose Quaternion-based Sliding Mode Control (QASMC) to
enhance the stability and mitigate chattering issues. The sliding
surface is modified to avoid singularity compared to conventional
SMC. To improve the robustness of controllers, the
control parameters are updated using adaptation laws. Furthermore, the asymptotic stability of translational and rotational
dynamics is guaranteed by utilizing Lyapunov stability and
Barbalet Lemma. Finally, the comprehensive comparison results
are provided to verify the effectiveness of the proposed controllers
in the presence of unknown time-varying parameter uncertainties
and significant initial errors. 
\end{abstract}


\section{Introduction}\label{sec1}

\paragraph*{Motivation} In the last few decades, there has been a noticeable rise in developing drones for various applications intended for both civilian and research domains. Improvements in teleportation technology, energy store capacity, electronics, and computational theory achieve such an advancement \cite{hashim2023exponentially,alzahrani2020uav,ghamari2022unmanned,ali2024deep,hashim2024uwb}. Among various kinds of Unmanned Aerial vehicles (UAVs), quadrotor has garnered considerable attention and has been used in numerous projects. The streamlined and straightforward mechanical configuration of quadrotors renders it a suitable candidate to operate in both indoor and outdoor environments \cite{hashim2023exponentially,hashim2024uwb,hashim2021geometric}.  

\paragraph*{Related Work} Designing controllers for UAVs, including quadrotors, is challenging regarding control design, model representation, and stability \cite{hashim2023observer,jonnalagadda2024segnet}. Various control approaches have been employed to address the control system limitations of UAVs. In the literature, several control strategies have been suggested to enhance the stability, guidance, and navigation of UAVs. Feed-forward and feedback control methods have been suggested to manage the linearized quadrotor model \cite{formentin2011flatness}.
The work in \cite{bouabdallah2004pid} and \cite{okyere2019lqr} used conventional PID and LQR-based controllers to enhance stability. To improve robustness, optimal control methods, including $H_{\infty}$ approaches, were utilized in \cite{kim2019comprehensive}. However, it is noteworthy that such controllers were designed for the linearized model of UAVs, which might not accurately represent real-world scenarios. In this context, the work in \cite{mathisen2021precision} addressed the nonlinear model of UAVs and designed Nonlinear Model Predictive Control (NMPC) in the presence of model mismatches. In the same spirit, Lyapunov-based nonlinear controllers using low-cost measurement units have been developed \cite{hashim2023exponentially,hashim2023observer}. The approaches in \cite{huo2022collision,xue2022compound} focused on the robustness aspect of controllers, utilizing advanced MPC techniques with optimal constraints. Despite MPC being an advanced method based on optimization and prediction, its main drawbacks in real-world applications are the computational burden and complexity. Among several suggested approaches, Sliding Mode Control (SMC) provides numerous advantages due to its inherent characteristics in handling model uncertainties, external disturbances, and nonlinearities \cite{shevidi2024quaternion}. SMC implementation not only guarantees stability but also reduces computational load. Designing a proper sliding surface with its control switch law is crucial in mitigating chattering, which is a significant reason for control failure.

The first and most challenging is the limitation of underactuated systems \cite{hashim2023observer,hashim2023exponentially}. Most UAVs are considered underactuated systems, where having control over all states is often impossible. While \cite{danesh2023auto} suggested super-twisting sliding mode control to enhance robustness and reduce chattering, they did not account for the practical and realistic limitations of underactuated UAVs. Backstepping control addresses the complexity of underactuated systems, leveraging its inherent cascaded strategy to generate auxiliary control\cite{kwan2000robust,shevidi2024quaternion}. Another challenge is the singularity issue, causing catastrophic instability in UAVs' flight missions. Preventing singularity in control design and model representation is fundamental for ensuring stability and effective path planning. While \cite{incremona2023guidance,bist2023state,zhao2023sliding,belmouhoub2023fast} introduced Euler angle-based sliding mode controllers to enhance control performance, they failed to develop singularity-free control designs and model representations. Most commonly established methods, such as Euler angle-based techniques, are susceptible to kinematic and model representation singularities. Such concerns hold profound significance when it comes to guaranteeing the stability and efficacy of path planning for UAVs.
\paragraph*{Contributions} Motivated by the above mentioned gaps, we propose Quaternion-based Adaptive Backstepping Control (QABC) with its cascaded strategy, as an outer loop, to control the UAV translation components. Subsequently, we develop Quaternion-based Adaptive Sliding Mode Control (QASMC) for the attitude control system (inner loop). The QASMC is non-singular based on the proposed sliding surface. The control parameters are adaptively tuned to enhance the robustness and reduce chattering in control signals. Lyapunov theorem and Barbalet Lemma have been employed to guarantee asymptotic stability of the proposed approach. Finally, a comprehensive comparison has been simulated to validate the singularity-free and superior performance of the proposed control system.
\paragraph*{Structure}The remaining of the paper is organized as follows:
The preliminaries and notation are provided in Section \ref{sec:pre}. Section \ref{sec:dyna} formulates the UAV model in quaternion representation. Section \ref{sec:Control} proposes the novel quaternion-based controllers. Section \ref{sec:Result} presents the simulation results validating the performance of the proposed controller. Finally, Section \ref{sec:co} summarizes the paper.
\section{Preliminaries\label{sec:pre}}
The set of real numbers and the set of real numbers with dimensional space $\alpha$-by-$\beta$ are described by $\mathbb{R}$ and $\mathbb{R}^{\alpha\times \beta}$, respectively. $\boldsymbol{I}_{\alpha}\in\mathbb{R}^{\alpha\times \alpha}$ denotes identity matrix and $\boldsymbol{0}_{\alpha}\in\mathbb{R}^{\alpha\times \alpha}$ represents zero matrix.  $\|z\|=\sqrt{z^{\top} z}$ defines an Euclidean norm of the column vector $z\in\mathbb{R}^{n}$. Let UAV be traveling in 3D space where the moving-frame of the vehicle (body-frame) is represented by $\{\mathcal{B}\}$, meanwhile $\{\mathcal{I}\}$ denotes the inertial frame. A skew-symmetric matrix is given by $[\Omega]_\times$ such that:
\begin{align*}
	\left[\Omega\right]_{\times} & =\left[\begin{array}{ccc}
		0	 & -\Omega_{3} &\Omega_{2}\\
		\Omega_{3} & 0 & -\Omega_{1}\\
		-\Omega_{2} & \Omega_{1} & 0
	\end{array}\right],\hspace{1em}\Omega=\left[\begin{array}{c}
		\Omega_{1}\\
		\Omega_{2}\\
		\Omega_{3}
	\end{array}\right]
\end{align*}
Unit-quaternion offers a singularity-free orientation representation in 3D space. $Q$ denotes a unit-quaternion vector where $Q=[q_{0}, q_{1}, q_{2}, q_{3}]^{\top}=[q_{0},\boldsymbol{q}^{\top}]^{\top}\in\mathbb{R}^{4}$ and $||Q||=1$ with 
$q_{0}\in\mathbb{R}$ and $\boldsymbol{q}=[q_{1},q_{2},q_{3}]^{\top}\in\mathbb{R}^{3}$. The unit-quaternion elements are defined as \cite{hashim2019special}:
\[
Q=q_{0} + iq_{1}+ jq_{2} + kq_{3}
\]
where standard basis-vectors are represented by $i$, $j$, and $k$. The inversion of $Q$  is denoted as \cite{hashim2019special}:
\[
Q^{-1}=Q^{*}=[q_{0},-q_{1},-q_{2},-q_{3}]^{\top}=[q_{0},-\boldsymbol{q}^{\top}]^{\top}\in\mathbb{R}^{4}
\]
Consider the quaternion multiplication to be described by $\otimes$ and the two quaternion vectors describing true and desired UAV attitude are $Q=[q_{0},\boldsymbol{q}^{\top}]^{\top}$ and $Q_{d}=[q_{0d}, \boldsymbol{q}_{d}^{\top}]^{\top}$, respectively, such that $q_{0}, q_{0d}\in\mathbb{R}$ and $\boldsymbol{q}, \boldsymbol{q}_{d}\in\mathbb{R}^{3}$. The quaternion multiplication of true and desired vectors are given by \cite{hashim2019special}:
\[
Q\otimes Q_{d}=\left[\begin{array}{c}q_{0} q_{0d} + \boldsymbol{q}^{\top}\boldsymbol{q}_{d} \\ q_{0d}\boldsymbol{q} - q_{0} \boldsymbol{q}_{d} + [\boldsymbol{q}]_{\times}\boldsymbol{q}_{d}
\end{array}\right] = \left[\begin{array}{c} \tilde{q}_{0} \\ \tilde{\boldsymbol{q}}\end{array}\right]
\]
Quaternion multiplication is associative while it is not commutative. The UAV attitude representation in form of a rotational matrix described with respect to the Lie Group of the Special Orthogonal Group $SO(3)$ is described by \cite{hashim2019special}:
\[
R_{Q}=(q_{0}^{2} - \boldsymbol{q}^{\top}\boldsymbol{q})\mathbf{I}_{3} + 2\boldsymbol{q} \boldsymbol{q}^{\top} + 2 q_{0} [\boldsymbol{q}]_\times\in SO(3)\subset \mathbb{R}^{3\times 3}
\]
where $det(R_{Q})=+1$ and $R_{Q}R_{Q}^{\top}=\mathbf{I}_{3}$ with $det(\cdot)$ being determinant of a matrix and $R_{Q}\in\{\mathcal{B}\}$.
\section{Quaternion-based model representation of VTOL-UAV\label{sec:dyna}}
The quadrotors' architecture comprises four rotors that generate thrust by rotating blades in a particular direction. The flight motion of UAVs consists of two distinct parts, namely translation and rotation \cite{labbadi2019robust,hashim2023observer}. Most quadrotors are classified as under-actuated systems, requiring a cascaded control strategy to regulate the translational (external loop) and rotational (internal loop) motion. Translation motion generates thrust and the desired attitude for the rotation part. In this regard, finding a singularity-free model representation is crucial to design controllers. In this section, compared to other methods (e.g. Euler angles), the quaternion-based model representation of translation and rotation is addressed.
\subsection{UAV Quaternion-based Translational Dynamic}
The objectives of translation motion are to track desired position $P_{d}=[p_{d1}, p_{d2}, p_{d3}]\in\mathbb{R}^{3}$ and linear velocity $V_{d}=[v_{d1}, v_{d2}, v_{d3}]\in\mathbb{R}^{3}$. The translation dynamics of UAV is given by: \cite{hashim2023exponentially,hashim2023observer}:
\begin{equation}
	\label{eq:eq1}
	\text { Translation : }
	\begin{array}{ll}
		\left[\begin{array}{l}
			\dot{P}\\
			\dot{V}
		\end{array}\right]=\left[\begin{array}{c}
			V \\
			g  e_z + m^{-1} R^{\top}_{Q} T
		\end{array}\right].
	\end{array}
\end{equation}
where $P=[p_x, p_y, p_z]\in\mathbb{R}^{3}$ and $V=[v_{1}, v_{2}, v_{3}]\in\mathbb{R}^{3}$ with $P,V\in\{\mathcal{I}\}$ are position and linear velocity of UAV, respectively. 
$m$ and $g$ denotes mass and gravity acceleration, respectively. $T=[0, 0, -\Im]^{\top} \in \mathbb{R}^3$ is the total control input, $\Im$ denotes the total thrust, and $e_z=[0, 0, 1]^{\top}$ is a standard basis vector of inertial frame. $R_{Q}$ is the quaternion rotation matrix as follows\cite{hashim2023exponentially,hashim2023observer}:
\begin{equation}
	\label{eq:eqR}
	R_{Q}=\left(q_{0}^2-\boldsymbol{q}^{\top} \boldsymbol{q}\right) I_{3}+2 \boldsymbol{q} \boldsymbol{q}^{\top}+2 q_{0} \mathbf{[q]}_{\times}
\end{equation}
where $q_{0}$ and $\boldsymbol{q}=[q_{1}, q_{2}, q_{3}]^{\top}$ are unit-quaternions and $[\boldsymbol{q}]_{\times}$ is Skew-symmetric matrix. By substituting \eqref{eq:eqR} in \eqref{eq:eq1}, the vehicle translation expanded state space model is obtained as\cite{hashim2023exponentially,shevidi2024quaternion}:
\begin{align}
	\label{eq:eqTranslation}
	\text {Translation:}  & \begin{cases}
		\dot{p}_x & = v_{1}\\
		\dot{v}_{1} & =-2 m^{-1 }\Im \left(q_{0} q_{2}+q_{1} q_{3}\right) \\ 
		\dot{p}_y & =v_{2}\\
		\dot{v}_{2} & =-2 m^{-1} \Im\left(q_{2} q_{3}-q_{0} q_{1}\right) \\
		\dot{p}_z & =v_{3}\\
		\dot{v}_{3} & =m^{-1} \Im\left(q_{1}^2+q_{2}^2-q_{0}^2-q_{3}^2\right)+ g
	\end{cases}
\end{align}

\begin{assum}
	\label{rem:rem1}
	The quadrotor is a rigid symmetric body, and the desired position trajectory $P_{d}$ is bounded and twice differentiable.
\end{assum}
To this point, the translation state space is defined by \eqref{eq:eqTranslation} and $T=[0, 0, -\Im]^{\top} \in \mathbb{R}^3$ as a total control input should be designed to track desired position and guarantee the stability of the translation dynamics by generating total thrust $\Im$ and desired values for rotational dynamics.
\subsection{Rotational Dynamic}
The rotational motion of quadrotors is defined by a highly nonlinear set of equations which are challenging in terms of control. $\mathcal{T}=[\tau_{1}, \tau_{2}, \tau_{3}]^{\top}\in\mathbb{R}^3$ denotes torque as an attitude control input. Desired trajectories are produced by a total thrust $\Im$. The rotational dynamics are as follows  \cite{hashim2023exponentially}:
\begin{equation}
	\label{eq:eqRotate}
	\text { Rotation}
	\left\{\begin{array}{ll}
		\dot{R}_{Q} & =-[\omega]_{\times} R_{Q} \\
		J_{m} \dot{\omega} & =[J_{m} \omega]_{\times} \omega+\mathcal{T},
	\end{array}\right.
\end{equation}
where $\omega\in\mathbb{R}^3$ is angular velocity and $R_{Q}\in SO(3)$. Let $R_{Q}, \omega, \mathcal{T} \in\{\mathcal{B}\}$, and $J_{m}=[J_{ij}]_{3\times3}\in\mathbb{R}^{3\times 3}$ is the inertia matrix which is diagonal.
By revisiting \eqref{eq:eqRotate} and substituting \eqref{eq:eqR}, the UAV rotational motion is rewritten as \cite{wu2022modeling,hashim2023observer}:
\begin{align}
	\label{eq:eqExRotate}
	\text{Rotation} & \begin{cases}
		\dot{q}_{0} & =-\frac{1}{2} \boldsymbol{q}^{\top} \omega\\
		\dot{\boldsymbol{q}} & =\frac{1}{2}\left(q_{0} \omega+[\boldsymbol{q}]_{\times} \omega\right)
		\\
		J_{m}\dot{\omega} & =[J_{m} \omega]_{\times} \omega+\mathcal{T}
	\end{cases}
\end{align}
where $Q=[q_{0}, q_{1}, q_{2}, q_{3}]^{\top}=[q_{0},\boldsymbol{q}^{\top}]^{\top}\in\mathbb{R}^{4}$. 
In the light of \eqref{eq:eqRotate}, and subsitituting $\boldsymbol{q}$, $J_{m}=diag(J_{11}, J_{22}, J_{33})\in\mathbb{R}^{3\times3}$, the state space expanded attitude model is expressed as:
\begin{align}
	\label{eq:eq8}
	\text{Rotation}&
	\begin{cases}
		\dot{\omega}_{1} & =J_{11}^{-1}(\left(J_{22}-J_{33}\right) \omega_{2} \omega_{3}+\tau_{1}) \\
		\dot{\omega}_{2} & =J_{22}^{-1}(\left(J_{33}-J_{11}\right) \omega_{3} \omega_{1}+\tau_{2} )\\
		\dot{\omega}_{3} & =J_{33}^{-1}(\left(J_{11}-J_{22}\right) \omega_{1} \omega_{2}+\tau_{3} )\\
		\dot{q}_{0} & =-\frac{1}{2}\left(q_{1} \omega_{1}+q_{2} \omega_{2}+q_{3} \omega_{3}\right) \\
		\dot{q}_{1} & =\frac{1}{2}\left(q_{0} \omega_{1}-q_{3} \omega_{2}+q_{2} \omega_{3}\right)\\
		\dot{q}_{2} & =\frac{1}{2}\left(q_{3} \omega_{1}+q_{0} \omega_{2}-q_{1} \omega_{3}\right) \\
		\dot{q}_{3} & =\frac{1}{2}\left(-q_{2} \omega_{1}+q_{1} \omega_{2}+q_{0} \omega_{3}\right)
	\end{cases}
\end{align}
where $\omega= [\omega_{1}, \omega_{2}, \omega_{3}]^{\top}$ and $\boldsymbol{q}=[q_{1},q_{2},q_{3}]^{\top}$ refer to angular velocity and quaternion elements, respectively. Also, $\tau_{1}$, $\tau_{2}$, and $\tau_{3}$ are rotational torque elements of UAV. 
To this end, the quaternion-based expanded nonlinear translational and rotational model dynamics are obtained in \eqref{eq:eqRotate} and \eqref{eq:eqTranslation}. Hence, in the next section, the input controllers are designed to guarantee stability and convergence of position, linear velocity, orientations, and angular velocity.      

\section{Quaternion-based control design\label{sec:Control}}
This section is dedicated to proposing a quaternion-based control system for underactuated UAVs. One way to tackle underactuated issues is to use a backstepping approach by generating virtual cascaded control. We will develop a Quaternion-based ABC for the UAV translational part (position and linear velocity) to handle uncertainties with its adaptive feature. Next, Quaternion-based Adaptive Sliding Mode Control (QASMC), as a straightforward approach with its streamlined implementation, will be developed for controlling the UAV rotational part (orientation and angular velocity). QASMC has low sensitivity to plant parameter variations and disturbances, eliminating the necessity of exact modelling. The adaptive feature of QASMC reduces chattering issues, and the sliding surface modification avoids singularity when compared to conventional SMC. The control is developed in a cascaded manner to address the underactuation challenge.
The outer loop of the proposed control system is QABC, which controls the UAV translational motion, while the inner loop is QASMC controlling the UAV rotational motion. The proposed schemes guarantee asymptotic stability and accurate trajectory tracking of both translational and rotational motion trajectories of the UAV.

\subsection{Adaptive Back Stepping control for translational motion\label{subsec:pose}}
In order to address the singularity shortcoming of Euler-based controllers (e.g. \cite{labbadi2019robust}), we design Quaternion-based Adaptive Backstepping Controller (QABC) as an effective alternative approach. Its recursive and cascaded implementation of QABC is a potential way to resolve underactuated system complexity. Let $P_{d}= [p_{d1}, p_{d2}, p_{d3}]$ be the desired position trajectories and let Assumption \ref{rem:rem1} be met. The first step to designing QABC is to compute the error between the true position trajectory and desired one as follows:
\paragraph*{Step 1} (Compute position and linear velocity errors) Based on \eqref{eq:eqTranslation}, position error of quadrotor $\tilde{P}=[\tilde{p}_{x},\tilde{p}_{y},\tilde{p}_{z}]^{\top}$ is obtained by:
\begin{equation}
	\label{eq:eq9}
	\begin{aligned}
		& \tilde{p}_{x}=p_{x}-p_{d1}\\
		& \tilde{p}_{y}=p_{y}-p_{d2}\\
		& \tilde{p}_{z}=p_{z}-p_{d3}\\
	\end{aligned}
\end{equation}
Then the linear velocity error is calculated by the derivative of \eqref{eq:eq9}:
\begin{equation}
	\label{eq:eq10}
	\begin{aligned}
		& \dot{\tilde{p}}_{x}=\dot{p}_{x}-\dot{p}_{d1}=v_{1}-\dot{p}_{d1} \\
		& \dot{\tilde{p}}_{y}=\dot{p}_{y}-\dot{p}_{2d}=v_{2}-\dot{p}_{d3} \\
		& \dot{\tilde{p}}_{z}=\dot{p}_{z}-\dot{p}_{3d}=v_{3}-\dot{p}_{d3}
	\end{aligned}
\end{equation}
The next step is to define the Lyapunov function candidate to guarantee stability and design virtual controllers necessary to generating the required thrust.
\paragraph*{Step 2}
Let us formulate controllers to ensure stability and accurate tracking based on the backstepping approach and real value functions as follows:
\begin{equation}
	\label{eq:eq11}
	\begin{aligned}
		& \mathcal{L}_{1}=\frac{1}{2} \tilde{p}_{x}^2, \hspace{20 pt} \mathcal{L}_{3}=\frac{1}{2} \tilde{p}_{y}^2, \hspace{20 pt} \mathcal{L}_{5}=\frac{1}{2} \tilde{p}_{z}^2 
	\end{aligned}
\end{equation}
Obtaining the derivative of \eqref{eq:eq11} and subsequently substituting \eqref{eq:eq10} in \eqref{eq:eq11}, one obtains:
\begin{equation}
	\label{eq:eq12}
	\begin{aligned}
		& \dot{\mathcal{L}}_{1}=\tilde{p}_x \left(\dot{\tilde p}_x\right)=\tilde{p}_x\left(v_{1}-\dot{p}_{d1}\right) \\
		&\dot{\mathcal{L}}_{3}=\tilde{p}_y \left(\dot{\tilde p}_y\right)=\tilde{p}_y\left(v_{2}-\dot{p}_{d2}\right) \\
		& \dot{\mathcal{L}}_5=\tilde{p}_z \left(\dot{\tilde p}_z\right)=\tilde{p}_z\left(v_{3}-\dot{p}_{d3}\right) \\
	\end{aligned}
\end{equation}
Hence, $v_{d1}$, $v_{d2}$, and $v_{d3}$ is defined to ensure the stability of \eqref{eq:eq12} as as follows:
\begin{equation}
	\label{eq:eq13}
	\begin{aligned}
		& v_{d1}=-\theta_{x} \tilde{p}_x+\dot{p}_{d1} \\
		& v_{d2}=-\theta_{y} \tilde{p}_y+\dot{p}_{d2} \\
		&v_{d3}=-\theta_{z} \tilde{p}_z+\dot{p}_{d3}
	\end{aligned}
\end{equation}
where $\theta_{x}$, $\theta_{y}$, and $\theta_{z}$ are positive constants. Based on the translation dynamics \eqref{eq:eqTranslation} and Assumption \ref{rem:rem1}, the position errors in \eqref{eq:eq9} reaches to zero by designing $v_{d1}$, $v_{d2}$, and $v_{d3}$ as in \eqref{eq:eq13}. One can conclude the stability and convergence of position error by computing the derivative of Lyupanov in \eqref{eq:eq11} and substituting \eqref{eq:eq13} in \eqref{eq:eq12}:
\begin{equation}
	\label{eq:eq14}
	\begin{aligned}
		\dot{\mathcal{L}}_{1} & =\tilde{p}_{x} \dot{\tilde p}_x=\tilde{p}_x\left(-\theta_{x} \tilde{p}_x+\dot{p}_{d1}-\dot{p}_{d1}\right) = -\theta_{x} \tilde{p}_{x}^2< 0 \\
		\dot{\mathcal{L}}_{3} & =\tilde{p}_{y} \dot{\tilde{p}}_{y}=\tilde{p}_{y}\left(-\theta_{y} \tilde{p}_{y}+\dot{p}_{d2} -\dot{p}_{d2}\right) = -\theta_{y} \tilde{p}_{y}^2< 0 \\
		\dot{\mathcal{L}}_5 & =\tilde{p}_{z}\dot{\tilde{p}}_{z}=\tilde{p}_{z}\left(-\theta_{z} \tilde{p}_{z}+\dot{p}_{d3}-\dot{p}_{d3}\right) = -\theta_{z} \tilde{p}_{z}^2< 0 
	\end{aligned}
\end{equation}
Since $\mathcal{L}_{1}>0$, $\mathcal{L}_{3}>0$, and $\mathcal{L}_{5}>0$ are positive and $\dot{\mathcal{L}}_{1}< 0$, $\dot{\mathcal{L}}_{3}< 0$, and $\dot{\mathcal{L}}_{5}< 0$ are negative definite, one can ensure that position errors converge to zero. To this end, the accurate tracking and stability of position is guaranteed. As a next step, we should show the precise tracking of linear velocity.

\paragraph*{Step 3} To guarantee the convergence of  $v_{1}$, $v_{2}$, and $v_{3}$ to $v_{d1}$, $v_{d2}$, and $v_{d3}$ defined in \eqref{eq:eq13}, we compute the error between true and actual values as follows:
\begin{equation}
	\label{eq:eq15}
	\begin{aligned}
		& \tilde{v}_{1}=v_{1}-v_{d1} \\
		& \tilde{v}_{2}=v_{2}-v_{d2} \\
		& \tilde{v}_{3}=v_{3}-v_{d3} 
	\end{aligned}
\end{equation}
Similar to \textit{Step 1} and \textit{Step 2}, let us define the following Lyapunov function candidate:
\begin{equation}
	\label{eq:eq16}
	\begin{aligned}
		&\mathcal{L}_{2}=\mathcal{L}_{1}+\frac{1}{2} \tilde{v}_{1}^2\\
		& \mathcal{L}_4=\mathcal{L}_{3}+\frac{1}{2} \tilde{v}_{2}^2\\
		& \mathcal{L}_6=\mathcal{L}_5+\frac{1}{2} \tilde{v}_{3}^2
	\end{aligned}
\end{equation}
One method to tackle the problem of underactuation involves utilizing the backstepping approach by generating auxiliary control (virtual control). Now, $F_x$, $F_y$, and $F_z$ are introduced as auxiliary commands (inner controls) aimed at producing the thrust. Plugging \eqref{eq:eq10}, 
\eqref{eq:eq9}, and \eqref{eq:eq13} in \eqref{eq:eq16} and by $F_{x}$, $F_{y}$, and $F_{z}$ as virtual controls, the following equations are derived as:
\begin{equation}
	\label{eq:eq18}
	\begin{aligned}
		\dot{\mathcal{L}_{2}} & =-\theta_{x} \tilde{p}_x^2+\tilde{v}_{1}\left(F_{x}+\left(\theta_{x} \dot{\tilde{p}}_x-\ddot{p}_{d1}\right)\right) \\
		\dot{\mathcal{L}_{4}} & =-\theta_{y} \tilde{p}_{y}^2+\tilde{v}_{2}\left(F_{y}+\left(\theta_{y} \dot{\tilde{p}}_{y}-\ddot{p}_{d2}\right)\right) \\
		\dot{\mathcal{L}_{6}} & =-\theta_z \tilde{p}_z^2+ \tilde{v}_{3}\left(F_{z}+\left(\theta_z \dot{\tilde{p}}_z-\ddot{p}_{d3}\right)\right) 
	\end{aligned}
\end{equation}
By plugging $\dot{\tilde{p}}_x$, $\dot{\tilde{p}}_y$, and $\dot{\tilde{p}}_z$, one can design the virtual control input $F_{x}$, $F_{y}$, $F_{z}$ as follows:
\begin{equation}
	\label{eq:eq19}
	\begin{aligned}
		& F_{x}=\theta_{x}^2 \tilde{p}_x-\hat{\psi}_{x} \tilde{v}_{1}+\ddot{p}_{d1}\\
		& F_{y}=\theta_{y}^2 \tilde{p}_y-\hat{\psi}_{y}\tilde{v}_{2}+\ddot{p}_{d2}\\
		& F_{z}=\theta_{z}^2 \tilde{p}_z-\hat{\psi}_{z} \tilde{v}_{3}+\ddot{p}_{d3}\\
	\end{aligned}
\end{equation}
Where $\hat{\psi}_{x}$, $\hat{\psi}_{y}$, and $\hat{\psi}_{z}$ are the estimation of adaptive control parameters for adjusting position control parameters enhancing stability and the speed convergence by online adaptation laws as:
\begin{equation}
	\label{eq:ad1p}
	\text{Update Law}  \begin{cases}
		\dot{\hat{\psi}}_{x}& = \eta_{1} \tilde{v}_{1}^2\\
		\dot{\hat{\psi}}_{y}& = \eta_{2} \tilde{v}_{2}^2\\
		\dot{\hat{\psi}}_{z}& = \eta_{3} \tilde{v}_{3}^2
	\end{cases}
\end{equation}
with $\eta_{1}$, $\eta_{2}$, and $\eta_{3}$ are positive user-defined constants.

\begin{thm}
	\label{thm:Theorem1}
	Consider the translation dynamics described by \eqref{eq:eqTranslation}. If the pose control is implemented as \eqref{eq:eq19} with adaptation law in \eqref{eq:ad1p}, then the closed-loop translational system is asymptotically stable, meanwhile position and velocity errors are $\tilde{P}\rightarrow 0_{3\times 1}$ and $\tilde{V}\rightarrow 0_{3\times 1}$, respectively.
\end{thm}

\begin{lem}
	\label{Lemm:lem1} \cite{labbadi2019robust}
	According to Barbalet Lemma,
	let $g(s)$ be a uniformly bounded continuous function, and suppose that
	$\lim_{s\to+\infty}{\int_{0}^{\top}g(s)\,ds }$ exists, then $g(s)$ converges to the origin asymptotically.
\end{lem}
\begin{proof}
	To prove the stability of subsystem and to define $\hat{\psi}_{x}$, $\hat{\psi}_{y}$, and $\hat{\psi}_{z}$,
	the Lyapoanov approach is used. Let Lemma. \ref{Lemm:lem1} hold true then by inserting controller \eqref{eq:eq19}, the stability of the translational system (both position and linear velocity) is guaranteed, based on Lyapunov function candidate as:
	\begin{equation}
		\label{eq:eq20}
		\begin{aligned}
			\mathcal{L}_{tot}	& = \frac{1}{2} \tilde{p}_x^2 + \frac{1}{2} \tilde{p}_y^2 + \frac{1}{2} \tilde{p}_z^2 + \frac{1}{2} \tilde{v}_{1}^2 + \frac{1}{2} \tilde{v}_{2}^2 + \frac{1}{2} \tilde{v}_{3}^2\\
			& +\frac{1}{2} \tilde{{\psi}}_{x}^{2} + \frac{1}{2} \tilde{{\psi}}_{y}^{2} + \frac{1}{2} \tilde{{\psi}}_{z}^{2}
		\end{aligned}
	\end{equation}
	where $\tilde{{\psi}}_{x}=\hat{\psi}_{x} - \psi_{x}$, $\tilde{{\psi}}_{y}=\hat{\psi}_{y} - \psi_{y}$, and $\tilde{{\psi}}_{z}=\hat{\psi}_{z} - \psi_{z}$.\\
	In view of \eqref{eq:eq18}, the derivative of \eqref{eq:eq20} is derived as:
	\begin{equation}
		\label{eq:d20}
		\begin{aligned}
			& \dot{\mathcal{L}}_{tot}= \dot{\mathcal{L}}_{2} + \dot{\mathcal{L}}_4 + \dot{\mathcal{L}}_6 + \tilde{{\psi}}_{x} \dot{\tilde{\psi}}_{x} +  \tilde{\psi}_{y}  \dot{\tilde{\psi}}_{y} + \tilde{\psi}_{z}  \dot{\tilde{\psi}}_{z}
		\end{aligned}
	\end{equation}
	By plugging \eqref{eq:eq19} in \eqref{eq:eq18} and \eqref{eq:d20}, one can obtain the simplification as follows:
	\begin{equation}
		\label{eq:eq20d}
		\begin{aligned}
			\dot{\mathcal{L}}_{tot} 	& = -\theta_{x} \tilde{p}_x^2 -\theta_{y} \tilde{p}_{y}^2-\theta_{z} \tilde{p}_{z}^2 -\hat{\psi}_{x} \tilde{v}_{1}^2
			-\hat{\psi}_{y} \tilde{v}_{2}^2\\
			& -\hat{\psi}_{z} \tilde{v}_{3}^2
			+ \tilde{\psi}_{x} \dot{\hat{\psi}}_{x} + \tilde{\psi}_{y} \dot{\hat{\psi}}_{y} + \tilde{\psi}_{z} \dot{\hat{\psi}}_{z} 
		\end{aligned}
	\end{equation}
	Then, by substituting \eqref{eq:ad1p} in  \eqref{eq:eq20d}, the stability of the translational dynamics is guaranteed as:
	\begin{equation}
		\label{eq:eq21}
		\begin{aligned}
			\dot{\mathcal{L}}_{tot} &= -\theta_{x} \tilde{p}_x^2 -\theta_{y} \tilde{p}_{y}^2-\theta_{z} \tilde{p}_{z}^2  -(\tilde{\psi}_{x}+\psi_{x}) \tilde{v}_{1}^2\\ 
			& -(\tilde{\psi}_{y}+\psi_{y}) \tilde{v}_{2}^2 
			-(\tilde{\psi}_{z}+\psi_{z})\tilde{v}_{3}^2 + \tilde{\psi}_{x} \tilde{v}_{1}^{2}+ \tilde{\psi}_{y} \tilde{v}_{2}^{2}\\
			& + \tilde{\psi}_{z} \tilde{v}_{3}^{2} 
			= -\theta_{x} \tilde{p}_{x}^2 -\theta_{y} \tilde{p}_{y}^2-\theta_z \tilde{p}_{y}^2
			- \psi_{x} \tilde{v}_{1}^{2}\\
			& - \psi_{y} \tilde{v}_{2}^{2}- \psi_z \tilde{v}_{3}^{2} \leq 0
	\end{aligned}\end{equation}
	Therefore, based on \eqref{eq:eq21}, $\dot{\mathcal{L}}_{tot}$ is negative definite, and $\mathcal{L}_{tot}$ is positive definite with $\dot{\mathcal{L}}_{tot}=0$ only if $\mathcal{L}_{tot}=0$. By doing so, it guarantees that the errors (both position error and linear velocity error) converge to the origin and the translational system is asymptotical stable.
	As \eqref{eq:eq21} becomes negative, it is ensured that the errors converge to zero and the system is asymptotically stable.\end{proof}
Now, one can compute the total thrust ($\Im$) to generate the desired attitude. From \eqref{eq:eqTranslation} and \eqref{eq:eq19}, $F_{x}$, $F_{y}$, and $F_{z}$, (auxiliary controls) defined in \eqref{eq:eq19} are rewritten as:
\begin{equation}
	\label{eq:eq22}
	\begin{aligned}
		& F_{x}=-2 m^{-1 }\Im\left(q_{0} q_{2}+q_{1} q_{3}\right)\\
		& F_{y}=-2 m^{-1} \Im\left(q_{2} q_{3}-q_{0} q_{1}\right)\\
		& F_{z}=-m^{-1} \Im\left(q_{0}^2-q_{1}^2-q_{2}^2+q_{3}^2\right)+ g\\
	\end{aligned}
\end{equation}
\paragraph*{Step 4} In order to find the actual thrust $\Im$, let us utilize \eqref{eq:eq22} as follows:
\begin{equation}
	\label{eq:eq23}
	\Im=m \sqrt{F_{x}^2+F_{y}^2+\left(F_{z}+g\right)^2}
\end{equation}
Based on \eqref{eq:eq23}, it is clearly evident that the absolute total of thrust force ($\Im$) will never exhibit singularity. In the light of \eqref{eq:eq22} and \eqref{eq:eq23}, the desired quaternions $q_{0 d}$, $q_{1 d}$, and $q_{2 d}$ are obtained as follows: For more information visit \cite{hashim2023exponentially, hashim2023observer}) :
\begin{equation}
	\label{eq:eq24}
	\begin{aligned}
		& q_{0d}=\sqrt{\frac{m\,\left(g+{F_{z}}\right)}{2\,{\Im}}+\frac{1}{2}}\\
		& q_{1 d}=-\frac{m\,{F_{y}}}{2\,{\Im}\,{q_{0 d}}}\\
		& q_{2 d}=\frac{m\,{F_{x}}}{2\,{\Im}\,{q_{0d}}}
	\end{aligned}
\end{equation}
It is assumed that $q_{3 d} = 0$. To this end, the design of QABC, tailored to the position system, is introduced. The controller serves a dual purpose: computing the total thrust and providing the required quaternion orientations for attitude control. Additionally, the Lyapunov stability criterion is employed to ensure the stability of the translational component.
\subsection{Robust quaternion-based sliding mode attitude control\label{subsec:atitude}}
In the previous section, the backstepping approach is used to tackle the underactuated problem of the UAV. The QABC controller in \eqref{eq:eq24} generates desired orientation $Q_{d}$ for attitude control. The proposed control system becomes complete once we design attitude control. The main objective of the following section is to design sliding mode control to ensure that the orientation and angular velocity are followed accurately. The attitude control steps are as follows:
\paragraph*{Step 5} Formulate quaternion-based attitude errors \cite{wu2022modeling,hashim2019special}:
\begin{equation}
	\label{eq:eq25}
	\begin{aligned}
		& \tilde{q_{0}} =q_{0}q_{0d}+\boldsymbol{q}_{d}\boldsymbol{q}^{\top}\\
		& \tilde{q} =q_{0d}\boldsymbol{q}-q_{0}\boldsymbol{q}_{d}+[\boldsymbol{q}]_\times \boldsymbol{q}_{d},\hspace{20pt} \tilde{q}= [\tilde{q}_1, \tilde{q}_2, \tilde{q}_3]^{\top}\\
		& \tilde{R}_{Q} =R_{Q} R_{d}^{\top}
	\end{aligned}
\end{equation}
$R_{Q}$ and $\tilde{\omega}$ are respectively quaternion rotation matrix and angular velocity error, defined by\cite{hashim2023exponentially}:
\begin{equation}
	\label{eq:eq26}
	R_{Q}=\left(q_{0}^2-\boldsymbol{q}^{\top} \boldsymbol{q}\right) I_{3}+2 \boldsymbol{q} \boldsymbol{q}^{\top}+2 q_{0} \mathbf{[q]}_{\times}
\end{equation}
\begin{equation}
	\label{eq:eq27}
	\tilde{\omega}=\omega -\tilde{R}_{Q} \omega_{d}
\end{equation}
$\omega_{d}$ is obtained through the following expression\cite{hashim2019special}:
\begin{equation}
	\label{eq:eq28}
	\boldsymbol{\dot{Q}_{d}}= (q_{0 d}\mathbf{I}_{3}+[\boldsymbol{q}_{d}]_{\times})\omega_{d}
\end{equation}
Using \eqref{eq:eq25}, 
\eqref{eq:eq26}, \eqref{eq:eq27}, one can find the attitude error dynamics  as\cite{wu2022modeling,hashim2023observer}:
\begin{equation}
	\label{eq:eq29}
	\begin{cases}
		\dot{\tilde{q}}_{0} & =\frac{1}{2} \tilde{\boldsymbol{q}}^{\top} \tilde{\omega}\\
		\dot{\tilde{\boldsymbol{q}}} & =\frac{1}{2}\left(\tilde{q}_{0} \tilde{\omega}+ [\tilde{\boldsymbol{q}}]_\times \tilde{\omega}\right)\\
		J_{m}\dot {\tilde{\omega} }& =
		[J_{m} \omega]_{\times} \omega+\mathcal{T}+\\
		&\hspace{0.3cm}J_{m}\left[\tilde{\omega}\right]_{\times} \tilde{R}_{Q} \omega_{d}-J_{m} \tilde{R}_{Q} \dot{\omega}_{d}
	\end{cases}
\end{equation}
\paragraph*{Step 6} Let us introduce the sliding surface for QASMC control design \eqref{eq:eq29}
\begin{equation}
	\label{eq:eq30}
	s=\gamma_{1}\tilde{\boldsymbol{q}}+ \tilde{\omega}
\end{equation}
where $\gamma_{1}$ is a positive constant. Then by computing the derivative of \eqref{eq:eq30}, the following equation is obtained as:
\begin{equation}
	\label{eq:eq31}
	\dot{s}=\gamma_{1}\dot{\tilde{\boldsymbol{q}}}+\dot{\tilde{\omega}}
\end{equation}
\begin{rem}
	Since the sliding surface is linear and it is susceptible to singularity, one can assume the singularity-free sliding surface  as follows:
	\begin{equation}
		\label{eq:eq32}
		\Bar{s}=\gamma_{1} \varrho(\tilde{\boldsymbol{q}})+\tilde{\omega}
	\end{equation}
	Where $\varrho(\tilde{\boldsymbol{q}})$ is denoted as\cite{hua2019fractional}:
	\begin{equation}
		\label{eq:eq33}
		\varrho\left(\tilde{\boldsymbol{q}}\right)= \begin{cases}\tilde{\boldsymbol{q}}^{\frac{c_{1}}{c_{2}}}, & \text { if } \Bar{s}=0 \text { or } \Bar{s} \neq 0,\left|\tilde{\boldsymbol{q}}\right|>\epsilon \\ \tilde{\boldsymbol{q}}, & \text { if } \Bar{s} \neq 0,\left|\tilde{\boldsymbol{q}}\right| \leqslant \epsilon\end{cases}
	\end{equation}
	with $c_{1}$ and $c_{2}$ are positive constant and  $0<\frac{c1}{c2}<1$. $\epsilon$ is a threshold small value.
\end{rem}
\paragraph*{Step 6} (Development of QASMC for the attitude control). Considering \eqref{eq:eq31} and \eqref{eq:eq32}, we develop the QASMC as follows:
\begin{equation}
	\label{eq:eq34}
	\begin{aligned}
		\mathcal{T}= & -J_{m} ([\tilde{\omega}]_{\times} \tilde{R}_{Q} \omega_{d} - \tilde{R}_{Q} \dot{\omega}_{d} ) - [J_{m} \omega]_{\times} \omega \\
		& - J_{m} \gamma_{1} \frac{c_1}{c_2} \Delta(\tilde{q}) -\mu_{1} s - \hat{\boldsymbol{\Lambda}} sign(s)
	\end{aligned}
\end{equation}
Where $\Delta(\tilde{q})=[\tilde{q}_1\dot{\tilde{q}}_1, \tilde{q}_2\dot{\tilde{q}}_2, \tilde{q}_3\dot{\tilde{q}}_3]^{\top}$. To enhance the robustness and decrease the chattering issue (the main reason for control failure), the adaptive law is used to update the parameter $\Lambda$, adjusting the control signal. $\mu_{1}$ is a random positive constant, and the adaptive online update law 
with $\boldsymbol{\hat{\Lambda}}$=$[\hat{\Lambda}_{1}, \hat{\Lambda}_{2}, \hat{\Lambda}_{3}]^{\top}$ is defined as:
\begin{equation}
	\label{eq:ad1}
	\text{Adaptive Law} \begin{cases}
		& \dot{\hat{\Lambda}}_{1}= \lambda ||s||\\
		& \dot{\hat{\Lambda}}_{2}= \lambda ||s||\\
		&  \dot{\hat{\Lambda}}_{3}= \lambda ||s||
	\end{cases}
\end{equation}
where $\lambda>0$ and no-zero positive number.
\begin{thm}
	\label{thm:Theorem2}
	Consider the quaternion-based attitude dynamics of UAV described in \eqref{eq:eq8}. The attitude system is asymptotically stable and the attitude error  $\tilde{Q}=[\tilde{q}_{0},\tilde{\boldsymbol{q}}^{\top}]^{\top}\rightarrow[1,0,0,0]^{\top}$ and $\tilde{\omega}\rightarrow 0_{3\times 1}$  in \eqref{eq:eq28} when the proposed control input \eqref{eq:eq34} with Adaptation update law \eqref{eq:ad1} is applied. 
\end{thm}	
\begin{proof}
	In order to prove Theorem.\ref{thm:Theorem2}, let Lemma 1. holds true, and the following Lyapunov function is chosen as:
	\begin{equation}
		\label{eq:eq35}
		\mathcal{L}_{s}=\frac{1}{2} s^{\top}s + \frac{1}{2 \lambda} \Tilde{\boldsymbol{\Lambda}}^{\top} \Tilde{\boldsymbol{\Lambda}}
	\end{equation}
	where $\Tilde{\boldsymbol{\Lambda}}= \hat{\boldsymbol{\Lambda}} - \boldsymbol{\Lambda}$. One can find the derivative of \eqref{eq:eq35} as:
	\begin{equation}
		\label{eq:eq36}
		\begin{aligned}
			\dot{\mathcal{L}}_s=s\left(\gamma_{1} \dot{\tilde{\boldsymbol{q}}} + \dot{\tilde{\omega}}\right)+ \frac{1}{\lambda} \Tilde{\boldsymbol{\Lambda}}^{\top} \dot{\tilde{\boldsymbol{\Lambda}}}
		\end{aligned}
	\end{equation}
	Based on \eqref{eq:eq28} and  substituting \eqref{eq:eq34}, the above equation is simplified as:
	\begin{equation}
		\label{eq39}
		\begin{aligned}
			&  \dot{\mathcal{L}}_s=- \mu_{1} ||s||^2 - s\hat{\boldsymbol{\Lambda}}sign(s) +
			\frac{1}{\lambda} \Tilde{\boldsymbol{\Lambda}}^{\top} \dot{\hat{\boldsymbol{\Lambda}}}
		\end{aligned}
	\end{equation}
	such that
	\begin{equation}
		\label{eq:ad2}
		\begin{aligned}
			\dot{\mathcal{L}}_s &=- \mu_{1} ||s||^2 - s(\tilde{\boldsymbol{\Lambda}}+\boldsymbol{\Lambda}) sign(s) +
			\frac{1}{\lambda} \Tilde{\boldsymbol{\Lambda}}^{\top} \lambda ||s||\\
			&  \leq - \mu_{1} ||s||^2 - \boldsymbol{\Lambda} ||s|| \leq 0
		\end{aligned}
	\end{equation}
	Thus, the attitude system \eqref{eq:eq8} is asymptotically stable with the proposed controller \eqref{eq:eq34}.\end{proof}
\begin{figure*}
	\centering{}\includegraphics[scale=0.44]{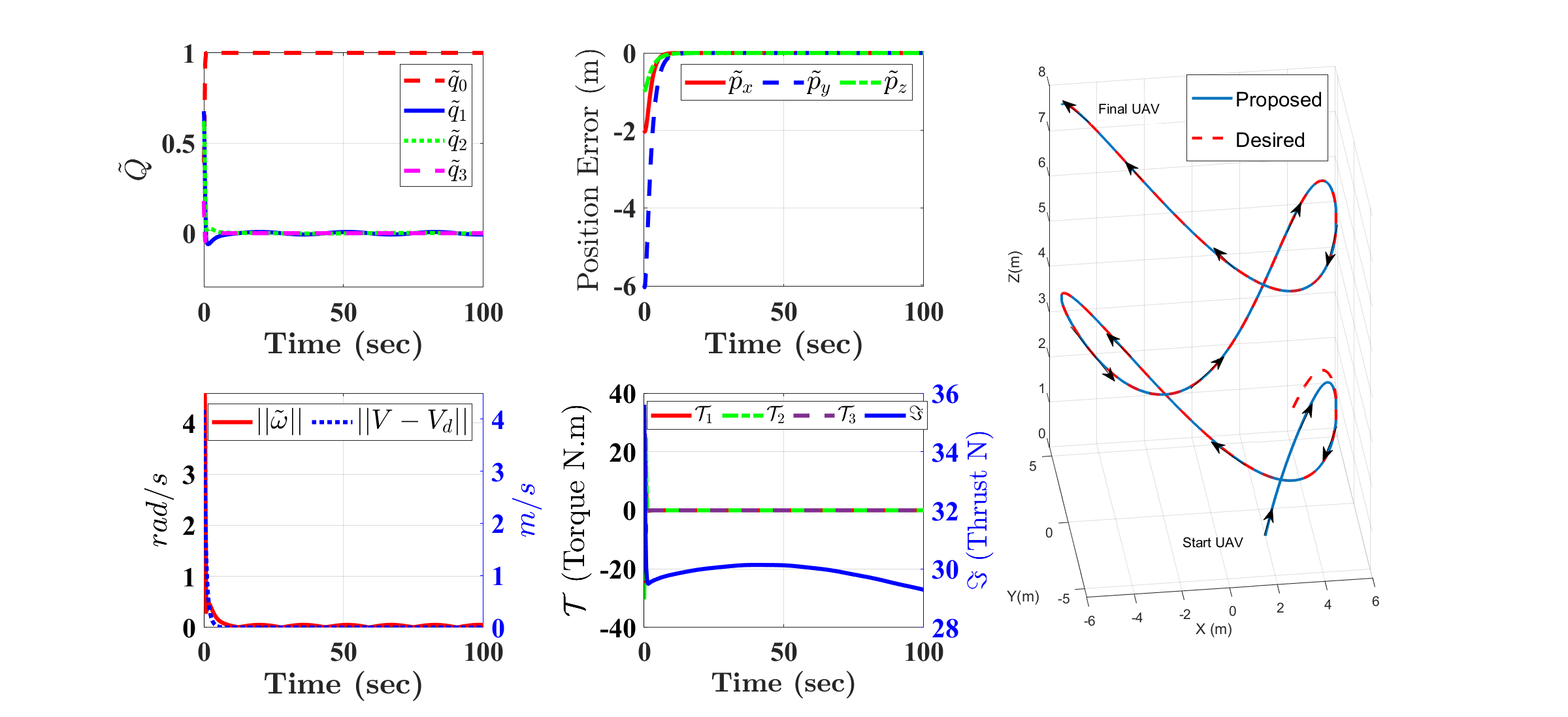}\caption{UAV flight trajectories, errors, and control signals}
	\label{fig:fig4}
\end{figure*}
\begin{figure*}
	\centering{}\includegraphics[scale=0.44]{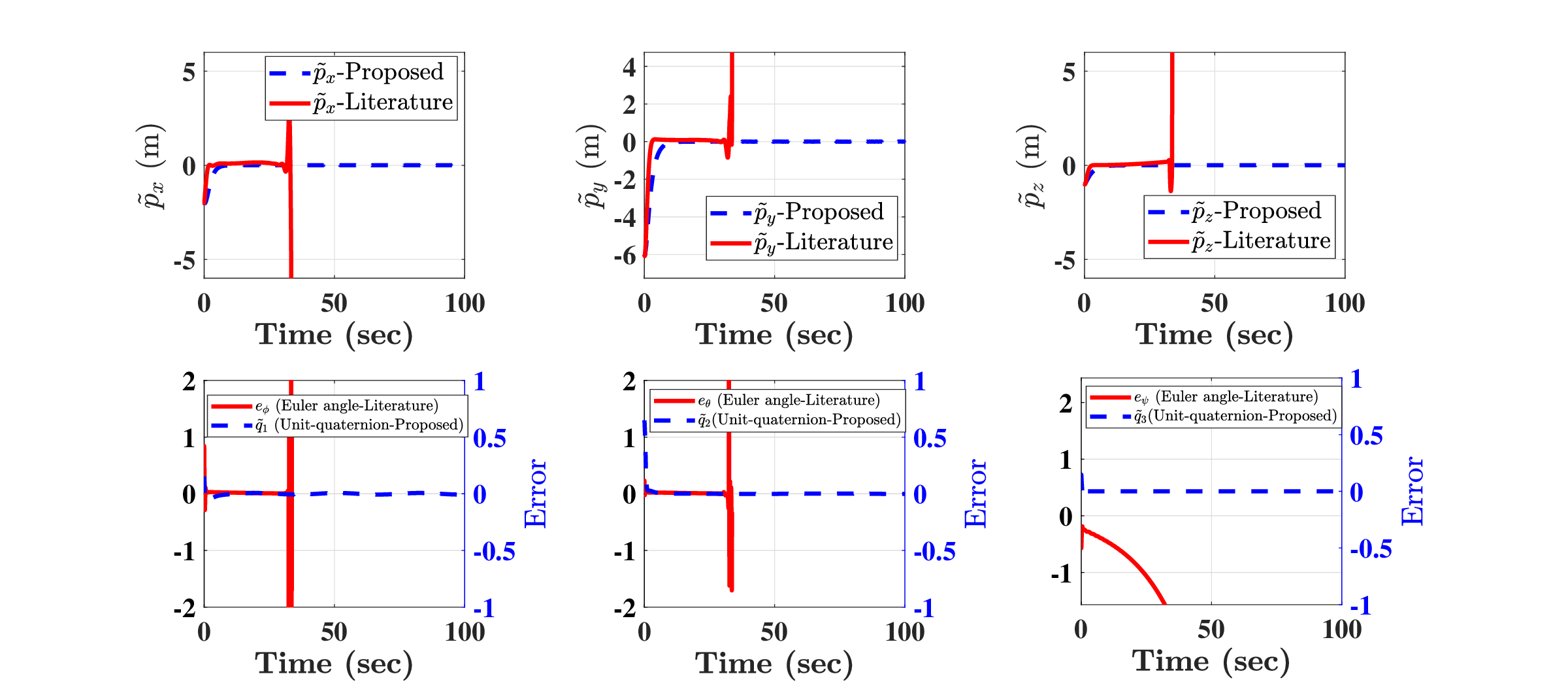}\caption{Position and attitude errors in literature\cite{labbadi2019robust} vs proposed work: Euler angles (literature\cite{labbadi2019robust} plotted in red solid line, while proposed approach plotted in blue dash-line.)}
	\label{fig:fig7}
\end{figure*}

\vspace{5pt}
\section{Results and discussion \label{sec:Result}}
This section presents the performance of the proposed quaternion-based position and attitude controllers for VTOL-UAV. The stability and tracking control capabilities are implemented to verify the proposed QABC and QASMC. Firstly, we simulate the proposed control system to study its performance in handling unknown time-varying parameters ($m$, $J_{11}$, $J_{22}$, $J_{44}$) uncertainties. Subsequently, a comprehensive comparison is provided to confirm the superiority and singularity-free of the proposed quaternion-based controllers compared to a commonly established method (Euler-based approach) in the literature.

\vspace{5pt}
\subsection{Discussion}
To evaluate the performance, a rigid symmetric body of UAV with mass $m$ and inertia $J_{m}= diag(J_{11}, J_{22}, J_{44})$ are assumed. The main goal is to track the desired position trajectory while the unknown time-varying uncertainties exist in parameters. The proposed control and UAV model parameters are provided in Table \ref{tab:Table3}.
The total time of simulation is set to 100 seconds. The proposed QABC is used as the translational control input to track the desired position and control linear velocity. Then, to overcome the underactuated issue of UAV, it generates the total thrust and desired attitude. Fig. \ref{fig:fig4} demonstrates the output performance of UAV controlled by the proposed QABC and QASMC. Fig. \ref{fig:fig4} confirms the robust and smooth performance of the proposed control system and its ability to track the UAV to the desired trajectory, despite the presence of unknown time-varying parameter uncertainties and large initial errors. The errors are depicted to converge from large initial errors to zero (equilibrium point). The rotational torque and thrust inputs are bounded, smooth, and robust in the presence of uncertainties, as shown in Fig. \ref{fig:fig4}.

\vspace{5pt}
\subsection{Comparision}
In this case, a comparison is conducted to demonstrate the singularity-free of the proposed quaternion-based controllers over the Euler angle-based controller approach in literature (e.g.\cite{labbadi2019robust}). The main deficiencies of the Euler angle-based controllers are kinematical and model representation singularities, causing instability and control failure. Fig. \ref{fig:fig7} compares attitude orientation tracking the performance of the proposed quaternion-based controllers with Euler angle-based ones. It confirms that while the proposed quaternion-based controllers have succeeded in following the desired orientations and positions, the Euler angle-based controllers fail to follow certain configurations and in turn UAV may become unstable, as shown in Fig. \ref{fig:fig7} which could potentially increase UAV incidents \cite{hashim2024Avionics}.

\vspace{5pt}
\section{conclusion\label{sec:co}}
In this work, we proposed a novel quaternion-based control system for the underactuated UAV. The proposed QABC controller tackled the underactuated complexity of controlling UAV by using the backstepping approach; meanwhile, the non-singular QASMC addressed stability and accurate tracking of attitude in the presence of unknown time-varying parameter uncertainties. The sliding surface was modified to provide the singularity-free feature for QASMC. The newly developed quaternion-based control system addressed the kinematic and model representation singularity. A comprehensive comparison was provided to confirm that the singularity-free of the proposed quaternion-based approach compared to other methods (e.g. Euler angle-based). The chattering and robustness shortcomings were addressed by designing adaptive features of the proposed controllers. Moreover, the asymptotic stability is guaranteed. Finally, the performance of controllers was verified by simulating underactuated UAV flight scenarios.

\vspace{5pt}
\begin{table}[t]
	\centering{}\caption{\label{tab:Table3}}
	
	\begin{tabular}{ll|ll}
		
		\hline Control Parameters & Value & System Parameters & Value\\
		\hline
		$\theta_{x}$ & 0.8 & m & 3.5kg\\
		$\theta_{y}$ & 0.5 & $J_{11}$ & $2kg.m^2$\\
		$ \theta_{z}$ & 0.4 & $J_{22}$ & $2kg.m^2$\\
		$\gamma_{1}$ & 10 & $J_{33}$ & $3.5kg.m^2$\\
		g & $9.8 m/s^2$\\
		
		$c_{1}$ & 3 \\
		$c_{2}$& 5 & \\
		\hline
		
	\end{tabular}
\end{table}

\newpage

\bibliographystyle{IEEEtran}
\bibliography{RefVTOL}

\balance

\end{document}